%% file: main.tex
\documentclass[10pt]{article} 
\usepackage[accepted]{tmlr}

\input{math_commands.tex}

\usepackage[sets]{cryptocode}
\usepackage{amsthm}
\usepackage{graphicx}
\usepackage{float}
\usepackage{booktabs}
\usepackage{hyperref}
\usepackage{url}

\newcommand{\lgc}[1]{}
\newcommand{\jtc}[1]{}
\newcommand{\scc}[1]{}
\newcommand{\rjc}[1]{}
\newcommand{\slc}[1]{}
\newcommand{\satya}[1]{}  
\newcommand{\Yehonathan}[1]{}

\newtheorem{theorem}{Theorem}
\newtheorem{lem}[theorem]{Lemma}
\newtheorem{definition}{Definition}
\newtheorem{remark}{Remark}
\newtheorem{corollary}{Corollary}[theorem]
\newcommand{\charlie}{\ensuremath{\mathcal{C}} }
\newcommand{\david}{\ensuremath{\mathcal{D}} }
\newcommand{\WCharlie}[1]{\bbw^{\charlie}_{#1}}
\newcommand{\WDavid}[1]{\bbw^{\david}_{#1}}
\newcommand{\bbw}{{\bf W}}
\newcommand{\bbu}{{\bf U}}
\newcommand{\bu}{{\bf u}}
\newcommand{\bv}{{\bf v}}
\newcommand{\ba}{{\bf a}}
\newcommand{\bc}{{\bf c}}
\newcommand{\br}{{\bf r}}
\newcommand{\bs}{{\bf s}}
\newcommand{\actv}[1]{\ba_{#1}}
\newcommand{\actvD}[1]{\ba^{\david}_{#1}}
\definecolor{darkgreen}{RGB}{0,100,0} 

\title{SLIP: Securing LLM's IP Using Weights Decomposition}

\author{\name Yehonathan Refael \email t-yrefael@microsoft.com \\
      \addr Microsoft\\
      \scriptsize{Author worked on this research during his internship at Microsoft}
      \AND
      \name Adam Hakim \email adamhakim@microsoft.com \\
      \addr Microsoft
      \AND
      \name Lev Greenberg \email levgreenberg@microsoft.com \\
      \addr Microsoft
      \AND
      \name Satya Lokam \email satya.lokam@microsoft.com \\
      \addr Microsoft
      \AND
      \name Tal Aviv \email talaviv@microsoft.com \\
      \addr Microsoft
      \AND
      \name Ben Fishman \email ben.fishman@microsoft.com \\
      \addr Microsoft
      \AND
      \name Shachar Seidman \email v-sseidman@microsoft.com \\
      \addr Independent Researcher
      \AND
      \name Racchit Jain \email t-racjain@microsoft.com \\
      \addr Microsoft
      \AND
      \name Jay Tenenbaum \email tenenbaumjay@microsoft.com \\
      \addr Microsoft
      \AND}

\begin{document}

\maketitle

\begin{abstract}
Large language models (LLMs) have recently seen widespread adoption, in both academia and industry. As these models grow, they become valuable intellectual property (IP), reflecting enormous investments by their owners. Moreover, the high cost of cloud-based deployment has driven interest towards deployment to edge devices, yet this risks exposing valuable parameters to theft and unauthorized use. Current methods to protect models' IP on the edge have limitations in terms of practicality, loss in accuracy, or suitability to requirements. In this paper, we introduce a novel hybrid inference algorithm, named SLIP, designed to protect edge-deployed models from theft. SLIP is the first hybrid protocol that is both practical for real-world applications and provably secure, while having zero accuracy degradation and minimal impact on latency. It involves partitioning the model between two computing resources, one secure but expensive, and another cost-effective but vulnerable. This is achieved through matrix decomposition, ensuring that the secure resource retains a maximally sensitive portion of the model's IP while performing a minimal amount of computations, and vice versa for the vulnerable resource. Importantly, the protocol includes security guarantees that prevent attackers from exploiting the partition to infer the secured information. Finally, we present experimental results that show the robustness and effectiveness of our method, positioning it as a compelling solution for protecting LLMs.
\end{abstract}

\jtc{TODO before TML submission: 
\begin{enumerate}
    \item Address comments from previous submissions.
    \item Who gets the final model inference? Charlie or David? Update in hybrid inference definition and protocol?
    \item Tie in to experiments, and make sure experiments match paper, not having a split but rather full verification. Specifically runtime is wrong.
    \item Problem: no Integrity for LLM (attention) in security paper!!!
    \item Distinguish between layer and Matrix-vector product for all sections, and use it correctly. The suggested solution solves specifically a matrix-vector product, and not a general layer (the non-matrix-vector-products are just done on Charlie).
    \item Make sure layer protocol is formal for attention, and it all makes sense.
    \item Make sure appendix theoretical runtime matches the new protocol where even in non-sensitive layers, the input is still naively k=0 masked.
    \item revert tmlr.sty to original
\end{enumerate}}

\input{introduction}

\input{framework_settings}


\input{SVD}

\input{protocol}

\input{experiments}

\input{discussion}

\bibliography{main}
\bibliographystyle{tmlr}

\appendix
\input{appendix}

\end{document}

%% file: math_commands.tex
\usepackage{amsmath,amsfonts,bm}

\def\eqref#1{equation~\ref{#1}}

\def\1{\bm{1}}

\DeclareMathAlphabet{\mathsfit}{\encodingdefault}{\sfdefault}{m}{sl}
\SetMathAlphabet{\mathsfit}{bold}{\encodingdefault}{\sfdefault}{bx}{n}



%% file: introduction.tex
\section{Introduction}
\label{sec:intro} 

Large Language Models (LLMs) are increasingly being commercialized due to their superior performance and vast applications \cite{zhao2023survey}. Alongside this trend, many researchers have been tackling challenges related to their deployment -- balancing costs, latency, performance and security. Companies which develop LLMs invest substantial capital to create these models. Not only in extensive computing resources for training, up to hundreds of millions of dollars \cite{perrault2024artificial}, but also on high-end talent, datasets and labeling, proprietary training procedures and unique network architectures. Consequently, the parameters of a pre-trained foundational model can be considered \emph{highly valuable} intellectual property (IP) for its owner, who would therefore be highly motivated to secure it from theft. 
Cloud services offer various security solutions against model theft and extraction attacks, however they often come with a high price tag with recent growth in demand \cite{businessinsiderChatGPTCould}.
As a result, there is a growing financial incentive to offload as much computational workload as possible from secure computing resources, like the cloud, to more cost-effective, but \emph{less secure} edge devices, such as servers and laptops.
This motivates us to develop a security solution that protects the model owner's intellectual property, despite the security gap in edge devices. \satya{We are focusing so much on \emph{large} language models, but our results are much more suitable and practical only for \emph{small} models and we don't have justification of our claims in the intro in our results for any real LLM's. This is a concern.}

Recently, \emph{hybrid inference} has emerged as a framework that aims to address various challenges related to deployment of LLM's \cite{8763885}. In this framework, tokens are generated using two computing resources that differ in their properties. Most strategies focus on cost and performance \cite{ding2024hybrid, bang2024crayon, 10.1145/3037697.3037698}, or on privacy \cite{cryptoeprint:2023/1147,chaopeng2023privacy, rathee2020cryptflow2}, rather than on model security. Few have aimed to protect model IP; For example, in \cite{schlogl2020ennclave} they propose splitting the model to a feature extractor and task specific classifier, securing the final layers of deep neural networks (DNNs) in an enclave, but their approach degrades the model accuracy, supports feed-forward DNNs only, requires model retraining, and does not have provable guarantees against information leakage. In contrast, \cite{pmlr-v202-zhou23h} use RL for model splitting, but their obfuscated parameters are inefficiently spread across the model, and they lack security guarantees against model restoration through observing inferences. Finally, \cite{li2024translinkguard} employ weight permutation and an authorization module on a trusted execution environment (TEE), but their approach is vulnerable to hardware attacks, impractical for consumer devices due to the cost of TEEs, and not provably secure. \satya{What about Shadownet and more recent papers along similar lines? We need to mention them and compare our work to them.}

In this paper, we propose SLIP, a novel \emph{hybrid inference} approach to IP protection of the model that uses a cheap and low security computing resource in tandem with an expensive but secure resource. SLIP balances budget constraints, latency and security standards, but has no loss in model accuracy. We achieve this by offloading the majority of computations to the cost-effective but less secure computing resource, while protecting the most valuable yet computationally undemanding information on the secure but expensive resource. Moreover, we demonstrate that SLIP is applicable to various DNN architectures, including fully-connected multi-layer perceptrons (MLPs), convolutional neural networks (CNNs), and transformers.

\subsection{Our contributions}
We make the following contributions,
\begin{enumerate}
    \item We propose a \textbf{model decomposition} technique that strategically selects weight matrices across the model
    (as observed in \cite{meng2024pissa,wang2024svd,sharma2023truth}),
    and identifies the components within these matrices that embody the model IP, and safeguards them on the secure resource. Consequently, the remainder of the model contains most of the computations, but is essentially  useless for an attacker aiming to steal the model IP. 
    \satya{We should clearly say here that we have a companion paper that focuses on definitions and proofs of security and not go into this much detail on those. Instead, we must focus on the ML angle and experimental results. We must give more specific summary of the models, the experiments, latency overheads, performance metrics etc. }\jtc{I added a sentence in point 2 of the contributions, does this address your comment?}
    \item Since low and high security computing resources participate in each model inference, it incurs additional latency and introduces new potential attacks that could exploit the communication between them. Thus, our paper introduces a \textbf{secure inference protocol} to mitigate these security risks in hybrid inference. The protocol does not reveal the secure resource's intermediate outputs to the low-security resource, by \textbf{masking} these intermediate outputs with precomputed \textbf{uniformly distributed and input-independent} noise. Later, we efficiently cancel out this mask's effect without degradation in accuracy. We provide \textbf{theoretical proofs} for the correctness and security of the protocol, and in a companion paper~\cite{SLIP-security} provide even more formal definitions and proofs of the security of the protocol.
    \item We give \textbf{experimental evidence} for our method on LLMs. We show that our decomposition exposes a portion of the model with
    most of the computations but insignificant utility for an attacker aiming to steal the model IP. Moreover, we demonstrate that even when the attacker fine-tunes the exposed portion of the model in an attempt to restore its original performance, the resulting model is very weak compared to the original model, and the smaller the exposed portion, the weaker the resulting model.
\end{enumerate}

\subsection{Background and related works}
\textbf{Threat Model.} In our setting, we aim to protect from an attacker who aims to obtain the parameters, i.e., the IP \cite{10.1145/3595292} of our model. 
We assume the attacker knows the model architecture, has full knowledge of our protocol and any associated hyper-parameters \cite{petitcolas2023kerckhoffs}, and has indefinite physical access and full admin privileges in the low-security compute resource.\footnote{We treat this resource as if the attacker has possession of it, bypassed any existing security measures, has direct control over the exposed model parameters and can freely query the secure resource.}
Conversely, we assume that the secure resource is a trusted entity, inaccessible to the attacker except through the query API. 
Thus, our protocol focuses on protecting from an attacker who has access to the exposed model portion and attempts to restore the secret portion in the secure resource (or the complete model directly) through model restoration techniques \cite{chen2023lorashear,NEURIPS2023_44956951}.

\textbf{Model Extraction Attacks.} Model \emph{extraction} attacks are a specific type of security threat where attackers aim to create a duplicate or approximation of a machine learning model by using its public API (\cite{birch2023model,shamir2023polynomial,truong2021data,Kariyappa_2021_CVPR} see \cite{YAO2024100211,jagielski2020high} for review). Attackers manipulate the model's input-output pairs to learn about its structure and functionality, allowing them to create an approximation of the model. This can enable them to use the model for their own purposes, sidestep any usage restrictions or fees, or even uncover sensitive information about the original training data. Hence, this type of attacks could compromise compute resources at any security level, whether the inference system is hybrid or not, since it exploits the intended use of the service. Note that we focus on extraction attacks rather than \emph{stealing} \cite{10.1145/3595292}, since our setting is low-security resource edge devices.
 
\textbf{Protecting DNNs IP from Theft.} Several existing methods aim to secure DNNs. One obvious approach is \textit{Model Obfuscation}, the process of intentionally obscuring or hiding the sensitive information of a neural network while retaining its functionality. However, despite astonishing breakthroughs in recent years \cite{CACM-iO}, current solutions for cryptographically secure obfuscation either degrade model accuracy, or are extremely impractical even for simple functions \cite{CGM+-Obfus, cryptoeprint:2017/321}. Other approaches use knowledge distillation \cite{xu2018deepobfuscation}, architecture obfuscation \cite{zhou2022obfunas}, intentional model poisoning \cite{grailoo2022preventing}, increase model sensitivity to parameters perturbation \cite{szentannai2020preventing}, locate model tampering for user authorization \cite{10374926}, or various other means \cite{zhou2022model}. Additionally, passive protection techniques, such as \textit{watermarking} \cite{uchida2017embedding, adi2018turning, yan2023rethinking, yang2021robust,zhang2018protecting} (see \cite{li2021survey} for survey) or \textit{Attestation} \cite{chen2019deepattest}, help verify ownership and copyrights, but they cannot effectively \emph{prevent} unauthorized access and usage to demotivate attackers. Lastly, cryptographic techniques have been extensively applied in the area of Privacy Preserving Machine Learning (PPML). Such techniques include Homomorphic Encryption (e.g., \cite{acar2018survey}, \cite{onoufriou2021fully,lee2022privacy}),  Secure Multiparty Computation (e.g., \cite{knott2021crypten,Dalskov_2020, 10224651,10.1145/3133956.3134056, cryptoeprint:2023/1269}), and Differential Privacy (e.g., \cite{abadi2016deep}). However, in PPML, the emphasis is on user data privacy or training data rather than protecting model weights, and existing frameworks are not designed for this purpose, or do not account for the computational trade-offs between asymmetric parties with varying costs and security levels \cite{10.1145/3195970.3196023,juvekar2018gazelle, huang2022cheetah,rathee2020cryptflow2}.

%% file: framework_settings.tex
\section{General Framework: Model Decomposition and Hybrid Inference Protocol}
\label{sec:framework_settings}

We now describe our framework for secure hybrid inference, as detailed in Figure~\ref{fig:framework}. Our goal is to offload most of the inference workload to a low security (but computationally efficient) resource, while ensuring that the intellectual property of the model remains secure. To achieve this, we introduce a model decomposition scheme and a hybrid inference protocol.
We consider two entities:
\begin{enumerate}
    \item A \emph{secure} resource controlled by $\mathcal{C}$harlie ('C' as in Cloud), which has limited compute power but is trusted to keep sensitive data confidential.

    \item A \emph{low security} resource controlled by $\mathcal{D}$avid ('D' as in Desktop), which is computationally powerful but fully exposed to the attacker.

\end{enumerate}

We split the model parameters between these parties, then define a protocol for inference on any input that involves communication between them. The decomposition and protocol must jointly satisfy the desiderata of usefulness, safety, efficiency, and security, which we formally define in Section~\ref{sec:safe-decomp}.

\begin{figure}[t]
    \centering
    \includegraphics[scale= 0.5]{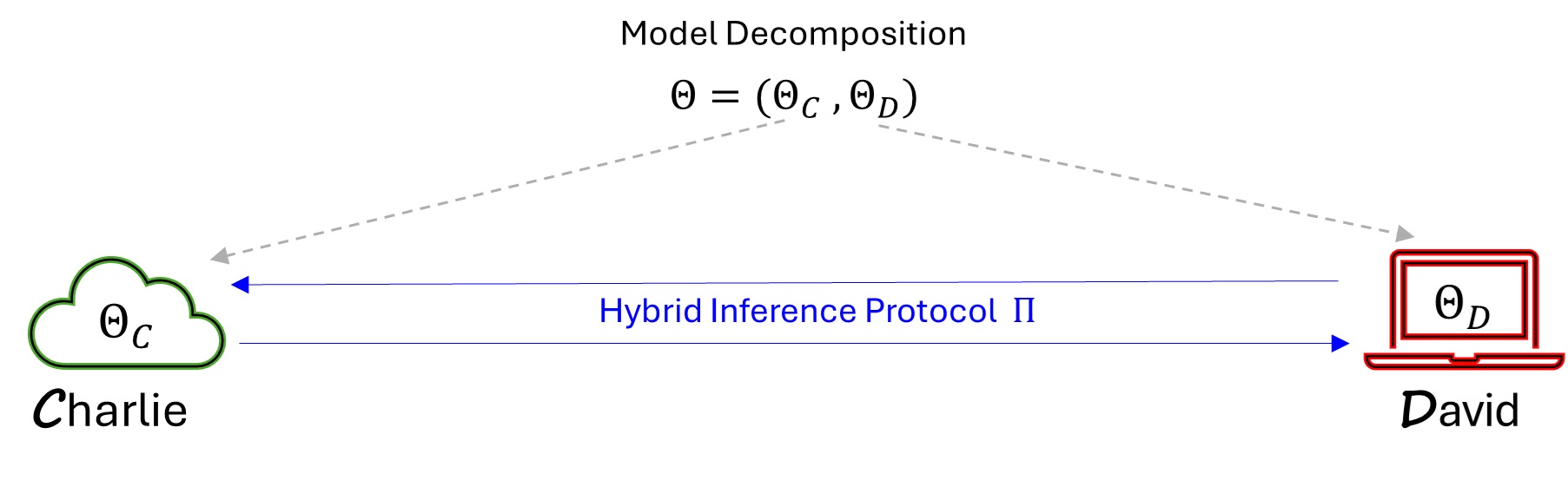}
    \caption{The proposed framework, consisting of the Model Decomposition and the Hybrid Inference Protocol}
    \label{fig:framework} 
\end{figure}

\begin{definition}[Model Decomposition]
Let $\phi_{\Theta}$ be a model with parameters $\Theta$. A \emph{model decomposition} is a pair
$(\phi_{\Theta_\charlie}, \phi_{\Theta_\david})$, where $\Theta=(\Theta_\charlie,\Theta_\david)$ denotes the parameter partition assigned to Charlie and David, respectively.
\end{definition}

\begin{definition}[Hybrid Inference Protocol] 
\label{def:hybrid-inf}
Given a model decomposition $\Theta = (\Theta_{\charlie}, \Theta_{\david})$ where Charlie and David only have access their respective portion of the model, a \emph{Hybrid Inference protocol} $\Pi := \Pi(\Theta_{\charlie}, \Theta_{\david})$ is an interactive protocol, that given an input $x$, jointly computes $\phi_{\Theta}(x)$.
\end{definition}

We now formalize what makes a good decomposition and protocol.

\subsection{What Makes a Good Decomposition and Protocol}
\label{sec:safe-decomp}

In order to protect a model's intellectual property (IP) while enabling efficient inference, we define four key properties that a decomposition and hybrid inference protocol should satisfy: \textbf{Usefulness}, \textbf{Safety}, \textbf{Efficiency}, and \textbf{Security}. These properties form the foundation of our framework.


We begin with intuitive explanations and follow with formal definitions.
\textbf{Usefulness} ensures the model IP isn't leaked through the David's portion of the model. \textbf{Security} means David doesn't learn anything about the full model beyond the model input/output pair during the hybrid inference protocol execution. \textbf{Safety} ensures that these input/output pairs are the major risk for David reconstructing the model (beyond David having his model portion).\footnote{Since a black-box adversary which sees sufficiently many input-output pairs of the model can recover the model $\Theta$ \cite{canales2024polynomial, carlini2024stealing}, and protecting from such attacks is a long term research challenge, we can only hope that our protocol does not give the attacker $\david$ much extra power beyond black-box access.} \textbf{Efficiency} means the protocol runtime for Charlie is much quicker than running the full model inference.

We encompass model IP through the notion of risk. For a model $\mathcal{M}$ and a loss function $\ell$, the model true risk is given by $\mathsf{risk}(\mathcal{M}) := \mathbb{E}_P \left[\ell(y, \mathcal{M}(x))\right],$ where $P(x,y)$ is the data distribution. We turn to the formal definitions.
\begin{definition}[Usefulness]
\label{def:useful}
Let $\phi_{\Theta}$ be the original model with parameters $\Theta$, and let $\Theta = (\Theta_{\charlie}, \Theta_{\david})$ be its decomposition. For an attacker algorithm $A$ (an adversarial David) that produces a model $A(\Theta_D)$ from $\Theta_D$, define its usefulness to be:
\begin{align}
\label{eq:usefulness}
\kappa_A := \frac {\mathsf{risk}(\phi_{\Theta})}{\mathsf{risk}(A(\Theta_D))}.
\end{align}
Note that lower $\kappa_A$ indicates the risk of the reproduced model is high, i.e., less useful. Also, for the decomposition in Section~\ref{sec:SVD}, the best $A(\cdot)$ just na\"ively returns $\Theta_D$. 
\end{definition}

\begin{definition}[Safety]
\label{def:safety} 
We say a decomposition $\Theta = (\Theta_\charlie, \Theta_\david)$ is \emph{safe} if,
for any Probabilistic Polynomial Time adversarial David $A$ that uses $\Theta_{\david}$ and black box access to $\phi_\Theta$ , there exists a Probabilistic Polynomial Time adversarial David $B$ using only black box access to $\phi_\Theta$, such that $B$'s reconstructed model has roughly the same risk as $A$'s reconstructed model.
\end{definition}



\begin{definition}[Security]
\label{defn:Security-decomp} 
A hybrid inference protocol $\Pi := \Pi(\Theta_{\charlie}, \Theta_{\david})$ is \emph{secure} if, on any input $x$, the \emph{view} of David during $\Pi$, beyond the input and the output, reveals no information about the weights of $\Theta_{\charlie}$.
\end{definition}

\begin{definition}[Efficiency]
\label{defn:eff-decomp} 
Let $T_{\Pi, \charlie}(\Theta_{\charlie})$ be the runtime of Charlie during the execution of a hybrid inference protocol $\Pi$, and let $T_{\charlie}(\Theta)$ be the runtime of Charlie for the full model inference. For $\epsilon>0$, we say $\Pi$ is \emph{$\epsilon$-efficient} if:
\begin{align}
\label{eq:eff-decomp}
T_{\Pi,\charlie}(\Theta_{\charlie}) & \; \leq  \; \epsilon \cdot T_{\charlie}(\Theta).
\end{align}
\end{definition}

In the following sections, we demonstrate how our proposed decomposition and protocol satisfy these properties both theoretically and empirically.

\subsection{Applying Our Framework to LLMs: A Roadmap}
\label{sec:offload-protect}
We now instantiate our framework to protect large language models (LLMs) during edge deployment. Our approach satisfies the desiderata introduced in Section~\ref{sec:safe-decomp} through two core components:
\begin{enumerate}
    \item \textbf{Model Decomposition}: We identify sensitive layers in transformer-based LLMs and apply singular value decomposition (SVD) to their weight matrices. The most informative singular components are kept on the secure side, while the remainder is offloaded.
    \item \textbf{Hybrid Inference Protocol}: To prevent the low security resource from recovering sensitive weights, we design a masking protocol that ensures secure interaction during inference.
\end{enumerate} 

In Section~\ref{sec:SVD}, we detail the decomposition strategy using SVD. In Section~\ref{sec:protocol}, we introduce and analyze the secure protocol. Section~\ref{sec:Empirical Experiments} evaluates the effectiveness of our full system on several LLMs, validating all four desiderata in practice.

%% file: SVD.tex
\section{Model decomposition via SVD}
\label{sec:SVD}

We now describe our method for securely decomposing the model parameters, satisfying the desiderata of usefulness and efficiency (Section~\ref{sec:safe-decomp}). Our strategy leverages a key empirical observation: in neural network layers which perform matrix-vector multiplications  (or layers which contain such operations, e.g., Linear Layers, Attention), the singular values of these matrices quickly decrease, as illustrated in Figure~\ref{fig:singular_values}. This suggests that most of the information of the matrix (corresponding to the IP within the layer) is concentrated in a few top singular components, so Charlie can hold these singular component, and efficiently compute the linear portion only using his components. This makes singular value decomposition (SVD) a natural tool for balancing efficiency with usefulness. 

\begin{figure}[t]
    \centering
    \includegraphics[width=1\linewidth]{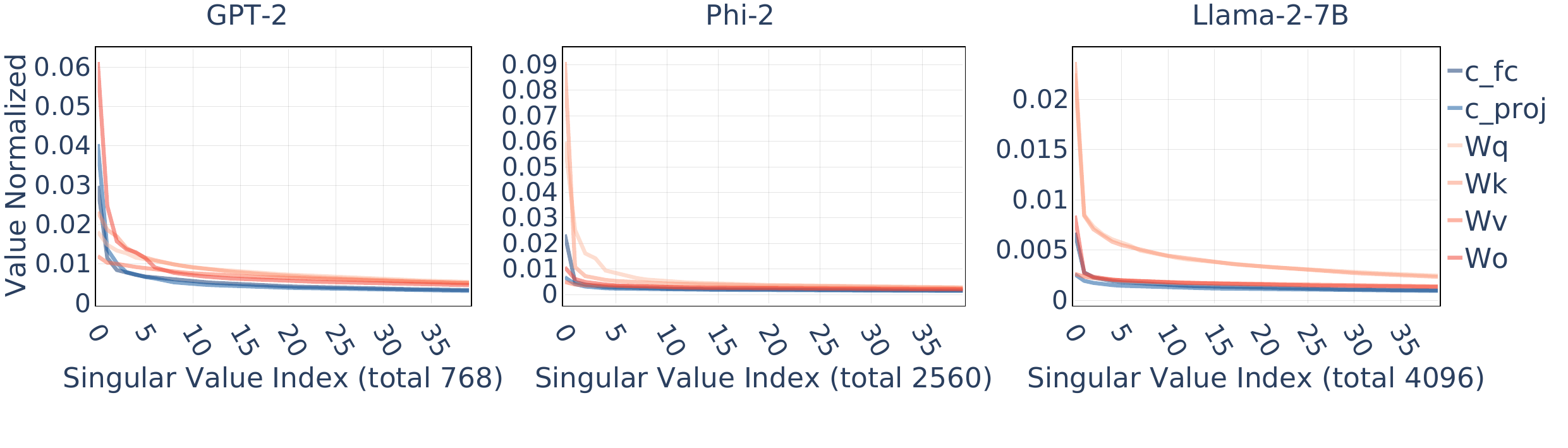}
    \caption{Singular values of the weight matrices of each model, quickly decreasing from largest to smallest.}
    \label{fig:singular_values} 
\end{figure}

\subsection{Layer-wise Decomposition}
\label{sec:layerwise-decomp}
Many layers in transformer-based LLMs and also other modern architectures, contain one or many matrix-vector multiplications, e.g., fully connected layers and attention layers (Appendix~\ref{app:attention}), and convolutional layers (Appendix~\ref{cnn}).
We apply SVD to decompose these matrices into two parts:
\begin{enumerate}
    \item a small number $k$ of top singular components, retained in the secure side (Charlie), and
    \item the remaining components, offloaded to the insecure device (David).
\end{enumerate}
For simplicity of presentation, to assume each layer has a single weight matrix, for attention layers we consider each matrix-vector operation within it as a separate linear layer for the decomposition above. 
Also for convolutions, we transform the convolution to its equivalent matrix-vector product representation. 

Formally, let  $\bbw\in\mathbb{R}^{m\times n}$ be a weight matrix. Its Singular Value Decomposition (SVD) is given by 
$$\bbw = \sum_{j=1}^r \sigma_{j} \bu_j \bv_{j}^{\top},$$
where $\{\sigma_j \in \mathbb{R}_{\geq 0}\}_{j \in [n]}$ are the singular values in decreasing order, $\{\bu_j \in \mathbb{R}^m\}_{j \in [m]}$ and $\{\bv_j \in \mathbb{R}^n\}_{j \in [n]}$ are the left and right singular vectors, and $r\leq \min\{m, n\}$ is the rank of $\bbw$. 
We define the sensitive portion $\WCharlie{}$ retained by Charlie to be the top $k$ singular components of $\bbw$, $$\WCharlie{ } = \sum_{j=1}^{k} \sigma_{j} \bu_{j} \bv_{j}^{\top},$$ and the remaining portion $\WDavid{ }$ offloaded to $\david$ to be the remaining components $$\WDavid{ } =\sum_{j=k+1}^{n} \sigma_{j} \bu_{j} \bv_{j}^{\top}.$$

This decomposition achieves two goals. It is useful, since the exposed portion $\WDavid{ }$ contains little signal, and it is efficient for David. Specifically, for a weight matrix $\bbw$, computing $\bbw x$ requires $O(mn)$ floating point operations, whereas if we store $\WCharlie{ }$ in its decomposed form, computing Charlie's part $\WCharlie{ }x$ requires only $O(k(n+m))$ operations.

We choose $k \ll r$, typically $k<10$, to minimize Charlie's compute, but $k$ must be large enough to ensure usefulness, specifically $k>1$ to prevent a trivial reconstruction attack described in Appendix~\ref{sec::MinimalNumberEigenvectors}.

In Section~\ref{sec:Empirical Experiments}, we empirically demonstrate the security and efficiency of this method, and show that it is far superior than other decomposition methods, such as offloading large norm columns or random singular component.

\subsection{Full-Model Decomposition}
To decompose a full model $\Theta$, we strategically select a set of “sensitive layers” among the layers defined by weight matrices $\bbw$, and decompose them as in Section~\ref{sec:layerwise-decomp} (see Appendix~\ref{Model Decomposition Strategy} for the selection process). The weight matrix $\WCharlie{}$ and the remaining computation within the layer (e.g., the non-linearity portion of a fully connected layer) are assigned to $\Theta_\charlie$, and $\WDavid{}$ is assigned to $\Theta_\david$. Non-sensitive layers are assigned entirely to $\Theta_\david$.

Formally, the overall decomposition becomes:
\begin{align*}
  &\Theta_\charlie=\{\WCharlie{i}\}_{i\in sensitive}\\
  &\Theta_\david=\{\WDavid{i}\}_{i\in sensitive}\cup \{{\bbw}^j\}_{j\notin sensitive}.  
\end{align*}

In Section~\ref{sec:Empirical Experiments}, we show that only decomposing a few layers, typically both at the beginning and end of the model using the SVD decomposition from Section~\ref{sec:layerwise-decomp} yields a useful and safe decomposition.

%% file: protocol.tex
\section{Secure Hybrid Inference Protocol}
\label{sec:protocol}
In Section~\ref{sec:SVD}, we saw how leverage SVD to decompose a model $\Theta$ into $\Theta = (\Theta_C, \Theta_D)$ such that Charlie's part $\Theta_C$ preserves most of the model's IP and David's part $\Theta_D$ contributes minimally to a malicious David which aims to extractact the original model $\Theta$. In this section, we present a \emph{hybrid inference protocol} where Charlie and David collaborate to compute the inference on this decomposed model by exchanging information on computations on their respective parts. We assume that Charlie owns the model input $x$ and is also responsible for producing the final output of the model.

\subsection{Why Na\"ive Hybrid Inference Leaks IP}

Consider a sensitive layer $i$ with weight matrix $\bbw_i$ decomposed as $\bbw_i = \WCharlie{i} + \WDavid{i}$. A na\"ive inference protocol to compute $\bbw_i \ba_{i-1}$ might proceed as follows. Charlie sends the input activation $\ba_{i-1}$ to David, receives $\WDavid{i}\ba_{i-1}$ from David, locally computes $\WCharlie{i} \ba_{i-1}$, and applies the nonlinearity $\sigma$ to obtain $\ba_i = \sigma(\WCharlie{i} \ba_{i-1} + \WDavid{i} \ba_{i-1})$. If David also observes $\ba_{i}$ (e.g., since the next layer $i+1$ is also sensitive) and if $\sigma$ is invertible (e.g., a sigmoid or a ReLU over non-negative inputs), he can reconstruct the linear equation $b_i = \WCharlie{i} \ba_{i-1} + \WDavid{i}\ba_{i-1}$, in which he knows all variables except $\WCharlie{i}$. By observing sufficiently many pairs $(\ba_i,\ba_{i-1})$, David can solve a system of linear equations to discover $\WCharlie{i}$, thereby leaking the sensitive information (IP) $\WCharlie{i}$ of the model $\Theta$ to David.  

 This attack shows that, without any additional security measures, the na\"ive inference protocol is inherently insecure. We next describe a security solution that protects against such leakage. 

\subsection{IP Protection via Masking}
\label{sec:ip_protection_via_masking}
Intuitively, to prevent information leakage about $\Theta_{\charlie}$ in the inference protocol, e.g., as in the attack above, we \emph{mask} Charlie's messages (internal model activations) so they appear random to David. The tricky part is to recover David's matrix-activation product $\WDavid{i}\ba_{i-1}$ from the output on David's part of the model on the masked input. Fortunately, since David's computation in sensitive layers is \emph{linear} (either in Linear layers or Attention layers which are a composition of linear operations), Charlie can later cancel the masks via \emph{precomputed mask cancellations} to recover David's matrix-activation product locally. Below we describe the main technical ideas of the protocol and its security proof. For more formal details of the protocol and proofs, see the companion paper~\cite{SLIP-security}. 

To realize this intuition in a concrete protocol, we assume that the model is quantized to integers, and our protocol uses modular arithmetic over \emph{integers mod $L$} for a sufficiently large $L$.
For each sensitive layer $i$, before inference, Charlie generates a uniformly random mask $\br_{i-1}$ of suitable length (as the layer input dimension) over integers $\bmod L$, and precomputes the corresponding cancellation mask $c_i:=\WDavid{i}\br_{i-1}$ and stores it.
During inference, for a layer input $\ba_{i-1}$, Charlie sends the masked input $\widetilde{\ba}_{i-1} = \ba_{i-1} + \br_{i-1} \pmod{L}$ to David, who computes $\widetilde{\ba_i^{\david}} := \WDavid{i} \cdot \widetilde{\ba}_{i-1} \pmod L$ and sends it to Charlie. In parallel, Charlie computes $\bbw_i^{\charlie} \ba_{i-1}$. Upon receiving $\widetilde{\ba_i^{\david}}$, Charlie then recovers the output $\bbw_i \ba_{i-1}$, using the precomputed cancellation mask $c_i$ via
\begin{equation*}
    \bbw_i \ba_{i-1}=\bbw_i^{\charlie} \ba_{i-1} + \widetilde{\ba_i^{\david}} - c_i, 
\end{equation*}
which holds since $\bbw_i = \bbw_i^{\charlie} + \bbw_i^{\david}$ and since $\widetilde{\ba_i^{\david}} - c_i = \bbw_i^{\david} {a}_{i-1} + \bbw_i^{\david} \br_{i-1} - c_i = \bbw_i^{\david} {a}_{i-1}$ by the definition of $c_i$. 
Finally, Charlie derives the next layer input by computing $\ba_i = \sigma(\bbw_i \ba_{i-1})$ from the recovered $\bbw_i \ba_{i-1}$.

To apply this intuition for Transformer models, for which the basic layer beyond Matrix-vector multiplication is attention, we give a formal and slightly more involved protocol in Appendix~\ref{app:attention}.

We now analyze the different properties of this protocol.

First, observe that the masking gives Charlie uniformly random values, independent of the activations.
\begin{theorem}[Perfectly secure masking]
    \label{thm:perfectly_secure_masking}
    Let a discrete random variable $s\in\mathbb{Z}^d$ and random noise $n\sim\bbu[0,L-1]^d$, and denote the masked variable by $s_n = \text{mod}(s + n, L).$ Then $s_n \sim \bbu[0,L-1]^d,$ and $s_n$ and $s$ are independent.
\end{theorem}


Next, we observe that the computation is correct (assuming David follows the protocol), 
\begin{lem}
\label{lem::remove_noise}
Assuming $L \geq \|\bbw_i^{\david} \ba_{i-1}\|_{\infty}$, Charlie's computation correctly computes $\bbw_i \ba_{i-1}$, over the integers modulo $L$.  
\end{lem}

Finally, note that this protocol is efficient since Charlie performs a constant number of vector addition, subtraction and sigmoid computations, but they are dominated by the two heavy operations: \textbf{1.} $\WDavid{i} \br_{i-1}$, which is independent of the input and is therefore pre-computed before the protocol, and \textbf{2.} $\WCharlie{i}\ba_{i-1}$ which is efficient since $\WCharlie{i}$ contains only the top $k$ SVD components (see Section~\ref{sec:layerwise-decomp}).

\subsection{Full Model Inference Protocol \texorpdfstring{$\Pi$}{Pi}}\label{subsec:full-model-inference-protocol}
\jtc{Add visualization for full model}
To extend the protocols above from a single layer to a full model, the full model protocol $\Pi$ proceeds through the model layer by layer.\footnote{We observe that our solution can be extended to trade-off security with efficiency by splitting not the whole model, but rather parts of it which, informally, contain most of the model knowledge.}

For each sensitive layer, whose weights must be protected, we run the masking scheme from Section~\ref{sec:ip_protection_via_masking}.

For each non-sensitive layer which is fully exposed to David, we still need to mask (since otherwise David can leverage internal activations to recover the model easily using reconstruction attacks between activations), which is equivalent to running the masking and scheme for a layer decomposition in which Charlie retains $k=0$ components. 

For both sensitive and non-sensitive layers, all operations which are not Matrix-vector multiplications (e.g., nonlinearities, softmax, etc.) are computed completely by Charlie, so that Charlie can mask the next layer input. Note that the majority of computation is within Matrix-vector multiplications (or in attention layers, for which we give a modified efficient protocol per layer in Appendix~\ref{app:attention}), so our scheme is  efficient.


Our main result is that this full model protocol is secure. Specifically, we apply Theorem~\ref{thm:perfectly_secure_masking} for each layer separately to conclude that the masked messages from Charlie to David are independently and uniformly random vectors over integers $\bmod L$ of the corresponding length (even regardless of David's messages in the protocol). Hence Charlie's view of the protocol is just a set of independently and uniformly random vectors over integers (together with the model output which is shared), and the protocol is secure.\footnote{In our companion paper~\cite{SLIP-security}, we give a complete and formal proof for security.}
By the definition of safety, it follows that if the decomposition $\Theta = (\Theta_{\charlie}, \Theta_{\david})$ is safe, then the protocol does not leak the model IP to David. 


\subsection{Ensuring Integrity by Probabilistic Verification}
\label{sec:integrity_by_prob_verification}
The above protocol ensures security even against a \emph{malicious} David who could try to sabotage the inference protocol by sending incorrect or maliciously crafted messages to Charlie. However, the correctness of the protocol, i.e., that Charlie computes the model inference correctly, relies on the fact that David follows the rules of the protocol faithfully and returns the correct $\WDavid{i}\ba_{i-1}$ for each sensitive layer $i$. We present a mechanism for Charlie to detect such potentially malicious behavior by David and abort the inference protocol, such that any protocol which runs to completion, returns an identical output to the original model $\Theta$. We call this \emph{inference integrity.} Below, we describe how one could add inference integrity to the model IP protection protocol from the previous subsection, thus obtaining a \emph{fully secure} inference protocol even against a malicious David. Detailed technical definitions and security proofs for integrity appear in the companion paper \cite{SLIP-security}; here, we focus on sketching essential ideas, and present our formal protocols and experiments focusing on IP protection \textbf{without} integrity. 

The key idea\footnote{This is a well-known and easy to prove fact. But it is often attributed Freivalds who proved a more general claim about probabilistic verification of a matrix product.} relies on the fact that if a vector $\bv$ (over a finite field) is nonzero, then for a random vector $\br$, the inner product $\br^T \bv$ is nonzero with high probability. Charlie leverages this to validate that David's response $\widetilde{\ba_i^{\david}}$ satisfies $\widetilde{\ba_i^{\david}} =  \bbw_i^{\david} \cdot \widetilde{\ba}_{i-1} \pmod L$.  Specifically, Charlie checks whether $\br^T\widetilde{\ba_i^{\david}} \stackrel{?}{=} \br^T \bbw_i^{\david} \widetilde{\ba}_{i-1} \pmod L$. If this check fails, Charlie declares David as dishonest and aborts the inference protocol.  Note that since David's computation in our decomposition is linear (which is an associative operation), Charlie can \emph{pre-compute} the vector $\bs^T := \br^T \bbw_i^{\david}$ and store $\bs^T$. Hence, during inference, Charlie's check involves only checking whether $\br^T\widetilde{\ba_i^{\david}} \stackrel{?}{=} \bs^T\widetilde{\ba}_{i-1}$ (two inner products), which has negligible overhead compared to the matrix-vector product performed by David.
This pre-computation gives is a substantial savings during inference and can be hundreds to thousands times smaller for a typical linear layer compared to the cost of inference, rendering the protocol efficient even with integrity. 

The full model protocol with integrity simply adds integrity checks for each layer (i.e., for all Matrix-vector products in the inference).

\subsection{Putting things together: Secure Hybrid Inference on a Decomposition}
Now that we have a general method to decompose a general model (Section~\ref{sec:SVD}) and a general protocol to perform secure inference (Section~\ref{sec:protocol}), we select them to ensure the desiderata from Section~\ref{sec:safe-decomp}. Specifically, in Section~\ref{sec:Empirical Experiments} we find the layers which are most essential for the inference, and experimentally find the number of sensitive layers and the number of $k$ of retained singular components, to ensure the decomposition is both Useful and Safe. Then, efficiency follows thanks the low number of sensitive layers and $k$, together with our hybrid inference protocol which is communication and computation-efficient for Charlie. Finally, protocol security follows from Section~\ref{subsec:full-model-inference-protocol}. Combining these 4 desiderata, the solution is efficient for Charlie (efficiency), and David cannot reconstruct the model from his portion of the model (usefulness), nor from the interaction from the model which only gives him model input/output pairs (security), which don't help him beyond regular model black-box attacks (safety).






%% file: experiments.tex
\section{Experiments}
\label{sec:Empirical Experiments}
In this section, we evaluate the SLIP framework and explore different decomposition strategies using three LLMs: GPT-2 Small \cite{Radford2019LanguageMA}, Phi-2 \cite{microsoft2023phi2}, and LLaMA2-7B \cite{touvron2023llama}. We run experiments on an NVIDIA A100 GPU and use the WikiText benchmark \cite{merity2016pointer} for perplexity evaluation. Perplexity is computed using EleutherAI’s evaluation suite \cite{eval-harness}. Our experiments aim to assess the usefulness, safety, and efficiency of the decomposition and protocol.

Our first 4 experiments show how to find a safe and useful decomposition. Specifically, we first find which decoder blocks in the model are best to offload and select as sensitive (contain much IP), then find which layers in these blocks are best to offload, then how to decompose each such layer (SVD over the alternatives), and finally show the trade-off between offload $\%$ and security (perplexity). Our last experiment shows that our decomposition + protocol is efficient.

\textbf{Which Decoder Blocks are Sensitive.}
Recall that for efficiency, we want to identify the smallest set possible of sensitive layers which contain the model IP. For LLMs, the high-level layers are decoder blocks, each which contain internal layers (e.g., $\bbw_q$ for attention query, projection etc.). To identify sensitive decoder blocks, we consider sequences of consecutive decoder blocks with various sequences positions (indexes) and length (1,3,5 or 7), omit them from the model to simulate the portion moving completely to Charlie, and measure the usefulness via perplexity. Consistent with the findings of \cite{gromov2024unreasonable}, Figure~\ref{fig:Block Bypass} demonstrates that sensitive layers in the first decoder blocks damages David's model the most (i.e., most useful), in last ones slightly less, and the least in the middle.
Hence, we experiment with a class of decompositions for models with $\ell$, in which the first $a$ layers ($1$ to $a$), and the last $a$ layers ($\ell-a+1$ to $\ell$) are taken as sensitive, all with the same value $k$ (for simplicity). 
\begin{figure}[t] \centering\includegraphics[width=1\linewidth]{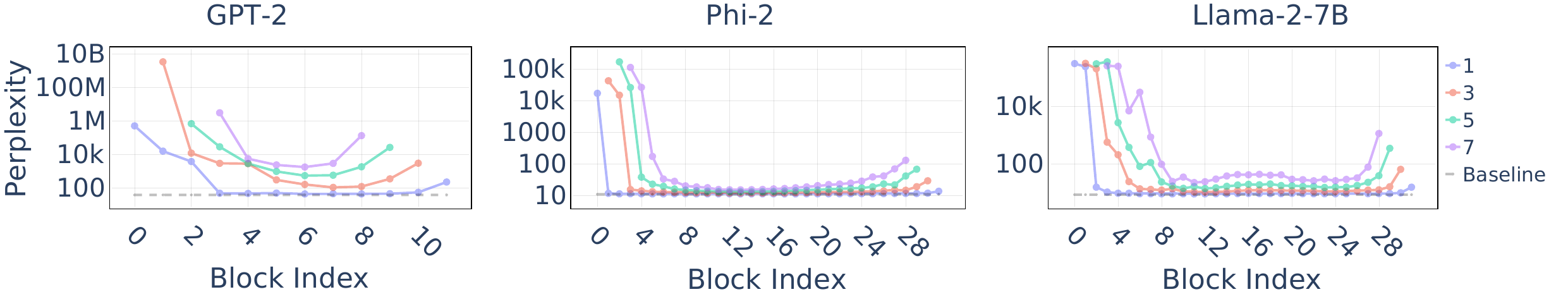}
    \caption{Performance degradation of the LLMs when bypassing windows of 1, 3, 5, or 7 consecutive decoder blocks, compared to the baseline model. The x-axis represents the center block of each window, while the y-axis shows the perplexity (lower perplexity is better). }
    \label{fig:Block Bypass}    
\end{figure}

\textbf{Which Decoder Block Internal Layers are Sensitive.}
Now that we know which decoder blocks should be taken as sensitive, we must decompose them, i.e., identify their internal sensitive layers. There are several different internal layers which contain Matrix-vector multiplication, specifically MLP layers ($\bbw_{c\_{fc}}$, $\bbw_{c\_{proj}}$) and attention layers ($\bbw_q$,$\bbw_k$,$\bbw_v$,$\bbw_o$). We empirically show that the decomposition usefulness greatly varies with the layer type above that we decompose, and the number of top singular components of the weight matrix retained by Charlie (usually $k=10$ is enough).
In Figure~\ref{fig:Decomposition Usefulness}, we plot the usefulness measured by the resulting perplexity of David's model as a function of both, where we average the results across all the blocks of the model of each layer type to see which is best for each model on average. We find that the taking the internal layer type $\bbw_{c\_{fc}}$ as sensitive is useful for GPT-2 and Phi-2, and the layer $\bbw_k$ is for Llama-2-7B.
\begin{figure}[t]
    \centering\includegraphics[width=1\linewidth]{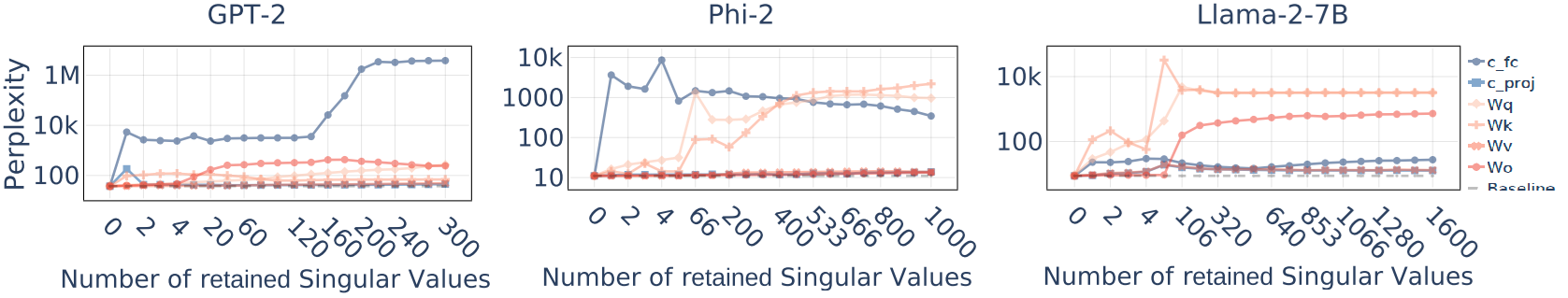}
    \vspace{.5em}
    \caption{Perplexity scores for all three models, after removing various numbers of singular components from different layer types in a single decoder block. Results are averaged across all possible blocks.}
    \label{fig:Decomposition Usefulness}        
\end{figure}

We now study Safety, that is, to what extent can an adversary David with his portion of the model and a set of model input/output pairs (taken from a public dataset), recover model performance (through fine-tuning \cite{han2024parameter} resources).
We use the public dataset Alpaca \cite{alpaca} and LoRA \cite{hu2021lora}, to fine-tune Phi-2 (or Llama-2-7B) for 10 epochs (to not overfit on the relatively small dataset), across five decomposition configurations and a baseline with randomized model parameters.

\textbf{How to Decompose a Sensitive Layer Safely}.
We compare our method Largest Singular Components (LSC) to four alternatives:
\begin{enumerate}
    \item \textbf{LCS} - retains the top $200$ singular components,
    \item \textbf{RCS} - randomly retains 200 singular components,
    \item \textbf{Norm CS} - retains 200 columns with largest L2 norm,
    \item \textbf{Random CS} - retains 200 random columns, and
    \item \textbf{FMS} - Offloads complete weight matrices to reach $97.5\%$ parameter offload to David.
\end{enumerate}

All methods are evaluated on decompositions of the same model, offloading 97.5\% of parameters and focusing on the 5 initial and 5 final decoder blocks, and after attempting to restore David's model using the Safety experiment mentioned above.
Table~\ref{table:alternatives} shows that the LSC method consistently results in the highest perplexity and lowest accuracy after fine-tuning David's model to the PPL wikitext and testing performance on 3 text classification tasks (ARC~\cite{arc-prize-2025}, Hellaswag~\cite{zellers2019hellaswag} and MMLU~\cite{peiyuan_liu_2023}), confirming it as the safest decomposition scheme.



\begin{table}[t] \centering\includegraphics[width=0.5\linewidth]{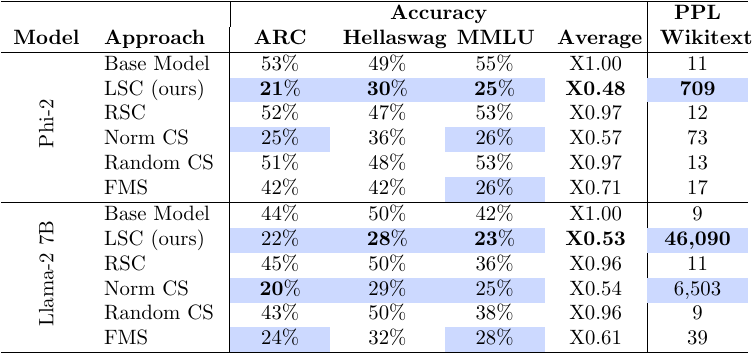}
    \caption{Benchmark results for each model and decomposition method on perplexity (PPL Wikitext) and classification accuracy (ARC, Hellaswag and MMLU) after restoration through fine-tuning. Blue indicates safe results while the bold is safest. LSC most often outperforms others.}
    \label{table:alternatives}    
\end{table}

\textbf{Efficiency vs Safety Trade-off}.
To experiment with various levels of efficiency, we consider decompositions created by varying the number $a$ of sensitive decoder blocks from the beginning and end of the model, and the number $k$ of top singular values retained by Charlie from each such layer (block). Each decomposition gives David a different fraction of the total model parameters, and we use this ratio as a proxy the inference time ratio. This proxy is accurate since Charlie's computation is dominated by matrix-vector multiplication in sensitive layers, and we store Charlie's portion in its decomposed form (See Section~\ref{sec:layerwise-decomp}).

Figure~\ref{fig:Model Recovery} shows the perplexity of a decomposed Phi-2 both after decomposition and after fine-tuning using the Safety study mentioned above. For all decompositions, Charlie's portion of the model has perplexity similar to a random weights model, and after retraining, the more computation offloaded to David, the lower the recovered perplexity. Specifically, configurations that offload 92\% to 97.5\% to David remain robust against recovery, i.e., safe.


\begin{figure}[t]
    \centering
    \begin{minipage}{0.1\linewidth} 
        \small 
        \vspace{10pt} 
        \setlength{\tabcolsep}{1pt} 
        \renewcommand{\arraystretch}{0.7} 
        \begin{tabular}{@{}lccc@{}}
            \toprule
            \textbf{Idx} & \textbf{Offload\%} & \textbf{Top Retained SVs} & \textbf{Sensitive Blocks} \\
            \midrule
            Rand & 0 & 0 & All \\
            1 & 85 & 380 & All \\
            2 & 92 & 200 & All \\
            3 & 95 & 200 & ±10 \\
            4 & 97.5 & 200 & ±5 \\
            5 & 99.4 & 50 & ±5 \\
            6 & 99.9 & 10 & ±5 \\
            \bottomrule
        \end{tabular}
    \end{minipage}
    \hspace{0.28\linewidth} 
    \begin{minipage}{0.6\linewidth} 
        \centering
        \vspace{0pt} 
        \includegraphics[width=0.67\linewidth]{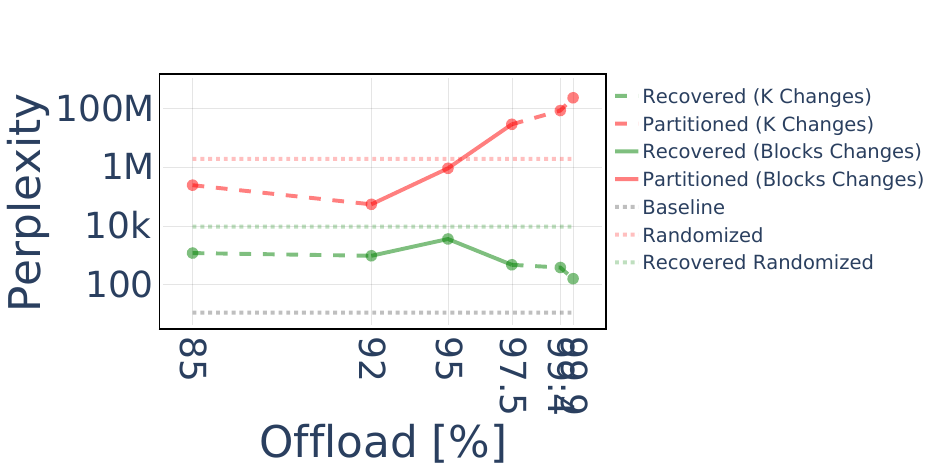}
    \end{minipage}
    \caption{Impact of fine-tuning on various decomposition configurations. Left: list of experimented configurations, with \% of offloaded computation to David, number of top singular values retained by Charlie from each layer, and the number of sensitive layers (blocks) in the beginning and end of the model. Right: the perplexity after decomposition (\textcolor{red}{red}) and after retraining (\textcolor{darkgreen}{green}). The experiment starting from randomized weights is dotted, and the baseline perplexity is in gray.}
    \label{fig:Model Recovery}
\end{figure}

\textbf{Efficiency Analysis.}
We simulate the latency of our approach decomposed to an edge device (Charlie) and a cloud server (David), considering the computational capacities of these hardware, network bandwith and delay.
Our simplified analysis, based on a version of Llama2 \cite{touvron2023llama} with only feed-forward layers, decomposed in a third of its layers, shows edge computation dominates latency (150.34 ms), followed by data transfer (71.31 ms), with cloud computation being minimal (0.66 ms). \jtc{While this does not include integrity, the efficiency is similar to that of the protocol for security.} Our approach is in fact $\varepsilon$-efficient for $\varepsilon=1.5\%$ (\jtc{without integrity}, see Appendix~\ref{sec_latency}), and one can theoretically prove that for a simplified feed-forward model with $\ell$ identical linear layers of type $\bbw \in \mathbb{R}^{d \times d}$, and $a$ sensitive layers in the beginning and the end with $k$ retained eigenvectors, the protocol is $O\left((2ak+\ell)/(\ell d)\right)$-efficient including both security and integrity.


%% file: discussion.tex
\section{Discussion}
\label{sec:Discussion}
In this paper we present SLIP, a novel hybrid inference technique designed to safeguard DNNs, particularly LLMs, against theft when offloaded to low-security resources like edge devices, leaving minimal computations on the secure resource, such as the cloud or a trusted execution environment. To our knowledge, SLIP is the first technique that is both practical for real-world use and provably secure, without compromising model accuracy (the hybrid inference protocol retains the original model outputs) or introducing significant latency (primarily influenced by network delays which can be mitigated in scenarios such as TEEs). It leverages the rapid decay of singular components in LLM weight matrices to obtain a useful, efficient and safe model decomposition, storing only the most sensitive singular components securely, while the remaining model is useless and robust against fine-tuning restoration, more than alternative partitioning approaches. 
To protect the secret portion of the model from the low-security resource, we give a provably secure and efficient inference protocol. These properties make SLIP highly promising for industry applications, enabling safe cost-effective model deployment on edge devices, and providing a foundation to further develop the hybrid inference framework. 
Although our approach is limited to DNNs that are quantized to integers, this is not a significant limitation since the majority of models are quantized pre-deployment to improve latency and memory usage. 

For future work, one can improve our framework applicability, and quantify its practical latency, obfuscate the model architecture, and improve the efficiency of the random mask generation and removal. Moreover, evaluating our method to other models and tasks is important. Finally, we believe that developing optimization algorithms to best balance the security, latency and computation offload of the decomposition, is very impactful. 

%% file: appendix.tex
\section{Proofs}
\label{app:proofs}

\subsection{Unmasking: Proof of Lemma~\ref{lem::remove_noise}}
\label{app:unmask}
Assuming $L \geq \|\bbw_i^{\david} \ba_{i-1}\|_{\infty}$, Charlie's computation correctly computes $\ba_i = \bbw_i \ba_{i-1}$, over the integers modulo $L$.  

\begin{proof}
\begin{align*}
\actvD{i+1}&=\text{mod}\left(\widetilde{\ba_{i+1}^{\david}} - \bc_i , L\right) \\
&=\text{mod}\left(\WDavid{i+1} \Tilde{\actv{i}} - \WDavid{i+1} \br_i  , L\right) \\
&=\text{mod}\left(\WDavid{i+1} \text{mod}(\actv{i} +\br_i,L) - \WDavid{i+1} \br_i  , L\right) \\
&\underbrace{=}_{(1)}\text{mod}\left(\WDavid{i+1} (\actv{i} +\br_i) - \WDavid{i+1} \br_i  , L\right) \\
&=\text{mod}\left(\WDavid{i+1} \actv{i}  , L\right) \\
&\underbrace{=}_{(2)}\WDavid{i+1} \actv{i} \\
\end{align*}
where, $(1)$ follows from properties of modulo operator and $(2)$ follows form the assumption that $\|{\WDavid{i+1}} \ba_i \|_{\infty} < L.$
\end{proof}

Recall \textbf{Theorem \ref{thm:perfectly_secure_masking}.}[Perfectly secure masking]
    Let a discrete random variable $s\in\mathbb{Z}^d$ and random noise $n\sim\bbu[0,L-1]^d$, and denote the masked variable by $s_n = \text{mod}(s + n, L).$ Then $s_n \sim \bbu[0,L-1]^d,$ and $s_n$ and $s$ are independent.

\begin{proof}
We first prove it for one-dimensional case $d=1$. 
It is enough to show 
$$P(x=a) = P( \text{mod}(s+n, L)=a) = 1/L.$$

We separate analysis to two distinct cases: \textbf{1.} $\text{mod}(s, L) \leq a,$ and \textbf{2.} $\text{mod}(s, L) > a$:
\begin{enumerate}
    \item For $\text{mod}(s, L) \leq a$, the equation $\text{mod}(s+n, L)=a$ holds iff  $n = a -s$. In this case,
   $$P( \text{mod}(s+n, L)=a) = P(n = a -s) = 1/L.$$
   \item Similarly, for $\text{mod}(s, L) > a$, equation $\text{mod}(s+n, L)=a$ holds iff  $n = L - (\text{mod}(s, L) -a)$. So also this case,
   $$P( \text{mod}(s+n, L)=a) = P(n = L - (\text{mod}(s, L) -a)) = 1/L.$$
\end{enumerate}
Finally, for any dimension $d>1$, by repeating the same analysis for each dimension, we get 
$$P(x=a) = \prod_{i=1}^d P(x_i == a_i) = 1/L^d,$$
where $x_i$ and $a_i$ are $i$-th coordinate of d-dimensional vectors $x$ and $a$. 
\end{proof}

\subsection{Minimal number of Singular Vectors that must be decomposed}\label{sec::MinimalNumberEigenvectors}
The following lemma shows that under the use of  SVD decomposition, more than one singular vector is needed to be removed from the model and stored separately on the cloud, noted by $W_c$. 

\begin{lem}[The minimal number of required singular vectors to be hidden]
    The matrix $V_c$ should contain more than one singular vector, in order to ensure that an attacker cannot reconstruct $W_{\text{c}}$ using only the singular vector available to them on the edge in $W_{\text{e}}$, or that they are combinatoricaly complex for the user to reconstruct. 
    \label{more_then_one_ev}
\end{lem} 
\begin{proof}

    Let us denote  \jtc{Isn't it usually U sigma V or Vt and not SVD?}
            \begin{align*}
            W_e &= S_eV_eD_e=\\
            &\left(\begin{array}{cccc}
            -- & S_1 & -- \\
             -- & S_2 & -- \\
             -- & S_3 & -- \\
             & \vdots & \\
             -- & S_n & --
            \end{array}\right)\cdot\left(\begin{array}{ccccc}
            \sigma_{1} & 0 & 0 & \cdots & 0 \\
            0 & \sigma_{2} & 0 & \cdots & 0 \\
            0 & 0 & \sigma_{3} & \cdots & 0 \\
            \vdots & \vdots & \vdots & & \vdots \\
            0 & 0 & 0 & \cdots & \sigma_{n}
            \end{array}\right)\cdot\left(\begin{array}{cccc}
            -- & D_1 & -- \\
             -- & D_2 & -- \\
             -- & D_3 & -- \\
             & \vdots & \\
             -- & D_n & --
            \end{array}\right)^t
            \end{align*}

    Let us suppose that one singular vector (without the loss of generally the largest one), namely $S_1$, should be enough to be designated for the matrix $W_c$, ensuring $W_{\text{c}}$ is unrecoverable to an attacker.. Then, in the edge device, we have $n-1$ orthonormal vectors in $\left\{S_2,S_3,\ldots,S_n\right\}$. Since all vectors in it are orthonormal it could be leveraged to construct the following $n-1$ equations,
    \begin{align*} 
    \centering
    S_2^t &\cdot S_1=0 \\
    S_3^t&\cdot S_1=0\\
    &\hspace{1cm}\vdots\\
    S_{n}^t&\cdot S_1=0,
    \end{align*}
    which results in only one degree of freedom. Now, simply applying the Gram-Schmidt process (or any other orthonormalization techniques) would result in the exact $S_1$ singular vector. This is in contradiction to the assumption.
\end{proof}

\begin{remark}[Finding the corresponding missing singular vector]
The sum of the singular values of a matrix equals the trace of that matrix, meaning $\operatorname{trace}(W_e)=\lambda_1+\lambda_2+\ldots+\lambda_{n-1}+\lambda_n$. Given $\lambda_2, \lambda_3, \ldots, \lambda_{n-1}$, the first singular value could be computed by
$$
\lambda_1=\operatorname{trace}(W_e)-\left(\lambda_2+\lambda_3+\ldots+\lambda_{n}\right)
$$
\end{remark}

\begin{corollary}
   Directly by Lemma~\ref{more_then_one_ev}, separating $n>>k\geq 2$ singular values and their corresponding singular vectors into their cloud designated matrices $S_c,V_c,D_c$, results in $k$ equations with $n$ unknown parameters. Thus, in this case, the level of freedom that can not be retrieved (by linear algebraic operations) can given by $n-k+1$.
\end{corollary}

\section{Practical Considerations} \label{sec::practicalApplication}
\subsection{Pseudo-random numbers generation}
\label{sec:PRNG}
In the main paper, we analyzed our protocol assuming true randomness from uniform distributions. However, in practice, we would use a Cryptographically Secure Pseudorandom Number Generator (CSPRNG).  A CSPRNG has the property that no polynomial time adversary can distinguish between true uniform distribution and the output of such a generator except with a negligible probability. By using such a generator, our security proofs will continue to hold. Otherwise, an adversary that breaks security of our schemes can be used to construct a distinguisher that breaks the security of a CSPRNG. 

\subsection{Model Decomposition Strategy}\label{Model Decomposition Strategy}
Technically, decomposition of an entire model begins with a list of triplets that contains indices of sensitive layers and defines how to perform the decomposition. In the case of LLMs, it would appear as follows:
\begin{align*}
&\left[ [ \mathbf{block}, \mathbf{layer\_type}, \mathbf{K} ]_1, \ldots, [ \mathbf{block}, \mathbf{layer\_type}, \mathbf{K} ]_n \right]
\end{align*}
where $\mathbf{block}$ is any decoder block within the LLM, $\mathbf{layer\_type}$ is any layer out of the possible layers in a decoder block, and $\mathbf{K}$ is the number of top singular components to hide in Charlie. Therefore, each triplet defines in which block and layer, and how many, singular components should be hidden. Using this list, an LLM can be decomposed into the secret model $\Theta_C$, comprised of the components detailed in the list, and an exposed version, comprised of all remaining components. 

Finding the optimal list with which to perform the partition is an optimization problem, that should consider trade-offs between the three following metrics: model security level, computational offload, and latency.  For example, increasing security by hiding a larger amount of singular components comes at the cost of less computational offload. Additionally, choosing more layers on which to perform decomposition results in additional networking latency. In this paper, we evaluated several decompositions, taking into account experimental findings in Section~\ref{sec:Empirical Experiments}. For example, we found it is useful to focus the first few and last few blocks of the decoder backbone, as removing them inflicts the most degradation in perplexity. Additionally, we found that sensitive layer types differ between LLMs, so we propose decomposing all layers in a selected block. Finally, we found that it is enough to focus on a small number of singular components (50 out of 2560). Therefore, we recommend that any attempts to apply our approach should begin with the two experiments outlined in Section~\ref{sec:Empirical Experiments}. These experiments are essential for identifying sensitive decoder blocks, key layer types, and the appropriate number of singular vectors for each model under consideration. Nonetheless, more optimized model decomposition could potentially be achieved by incorporating techniques such as Reinforcement Learning, Genetic Algorithms, or additional heuristics. We leave the exploration of these advanced methods to future work.

\section{Compute \& latency analysis details}\label{sec_latency}
\setcounter{equation}{0}
In this analysis, we evaluate the inference latency of a neural network model that is partitioned between an edge device and the cloud according to the protocol suggested in the paper. This setup necessitates considering both computational latency and data transfer latency between the two compute resources. To accurately measure the total latency, we break down the protocol into distinct phases and  calculate the computations in each, which contribute to the overall latency. Key parameters influencing latency include the bandwidth between the edge and cloud, the number of layers, the processor FLOPs capacity, and the size of the data being transferred. The formulas used for latency approximation are given by:

\textbf{Computation Time ($T_{\text{Compute}}$)}: The time required to perform the necessary computations on the edge or cloud.
\begin{equation}
    T_{\text{Compute}} = \frac{\text{FLOPs}}{\text{HardwareFLOP/s} \times \text{Utilization}}
\end{equation}

\textbf{Data Transfer Time ($T_{\text{Transfer}}$)}: The time required to transfer data between the edge and the cloud.
\begin{equation}
    T_{\text{Transfer}} = \lambda_{\mathrm{a}} + \frac{\text{input\_size} \times b \times s}{\text{Bandwidth}}
\end{equation}

Where $\lambda_{\mathrm{a}}$ is the network latency between the edge and the cloud, $b$ is the batch size, and s is the size of each activation value in bytes.
\begin{table}[H]
    \centering
    \vspace{1em}
    \caption{Description of phases and operations}
    \vspace{0.2in}
    \begin{tabular}{@{}llccc@{}}
        \toprule
        \textbf{Phase} & \textbf{Op Type} & \textbf{Frequency} & \textbf{Description} & \textbf{FLOPs/Data Size} \\
        \midrule
        \textbf{Upload Input} & Transfer & Once & \begin{tabular}[t]{@{}l@{}}
            Transfer input data $X$ of \\
            size $n$ to the cloud
        \end{tabular} & $n \times b$ \\
        \midrule
        \textbf{Edge-Only} \\ \textbf{Compute} & Compute & \begin{tabular}[t]{@{}l@{}}
            Each\\non-\\decomposed\\layer
        \end{tabular} & \begin{tabular}[t]{@{}l@{}}
            Compute all non- \\
            decomposed operations \\
            for matrices $W_{e_{i}}$
        \end{tabular} & \begin{tabular}[t]{@{}l@{}}
            $2 \cdot m \cdot n \cdot b+n$ \\
            $\cdot b$
        \end{tabular} \\
        \midrule
        \textbf{Edge-Partial} \\ \textbf{Compute} & Compute & \begin{tabular}[t]{@{}l@{}}
            Each\\decomposed \\
            layer
        \end{tabular} & \begin{tabular}[t]{@{}l@{}}
            Weights Matrix \\
            Multiplication \\
            $W_{e_{i+1}} a_{i_{\text{nois}}}$
        \end{tabular} & $2 \cdot m \cdot n \cdot b$ \\
        \midrule
        \textbf{Cloud-Partial} \\ \textbf{Compute} & Compute & \begin{tabular}[t]{@{}l@{}}
            Each\\decomposed \\
            layer
        \end{tabular} & \begin{tabular}[t]{@{}l@{}}
            Decomposed Matrices \\
            Multiplication 
            $U_{c} \Sigma_{c} V_{c} a_{i}$      
            \end{tabular}& \begin{tabular}[t]{@{}l@{}}
            $2 \cdot k \cdot b \cdot(m$ \\
            $+n)$
        \end{tabular} \\
        \midrule
        \textbf{Upload Edge} \\ \textbf{to Cloud} & Transfer & \begin{tabular}[t]{@{}l@{}}
            Each\\decomposed \\
            layer
        \end{tabular} & \begin{tabular}[t]{@{}l@{}}
            Transfer edge activation \\
            $a_{e, i+1_{\text{noisy}}}$ data to the \\
            cloud
        \end{tabular} & $n \times b$ \\
        \midrule
        \textbf{Cloud} \\ \textbf{Activation} \\ \textbf{Function} & Compute & \begin{tabular}[t]{@{}l@{}}
            Each\\decomposed \\
            layer
        \end{tabular} & \begin{tabular}[t]{@{}l@{}}
            Remove Noise + Vector \\
            addition + activation \\
            $a_{i+1}=\sigma\left(a_{c_{i}}+a_{e_{i}}-\tilde{\Delta}_{i}\right)$
        \end{tabular} & \begin{tabular}[t]{@{}l@{}}
            $2 \cdot n \cdot b \cdot\left(L_{v}\right.$ \\
            $+1)$
        \end{tabular} \\
        \midrule
        \textbf{Cloud Noise} \\ \textbf{Vector} \\ \textbf{Generation}& Compute & \begin{tabular}[t]{@{}l@{}}
            Each\\decomposed \\
            layer
        \end{tabular} & \begin{tabular}[t]{@{}l@{}}
            Sample $L_{v}$ noise vectors \\
            and create $\Delta_{i}$
        \end{tabular} & $2 \cdot n \cdot b \cdot L_{v}$ \\
        \midrule
        \textbf{Cloud Noise} \\ \textbf{Addition} & Compute & \begin{tabular}[t]{@{}l@{}}
            Each\\decomposed \\
            layer
        \end{tabular} & \begin{tabular}[t]{@{}l@{}}
            Vector addition $a_{i_{\text{noisy}}}=$ \\
            $a_{i}+\Delta_{i}$
        \end{tabular} & $n \cdot b$ \\
        \midrule
        \textbf{Activation Data} \\ \textbf{Download} & Transfer & \begin{tabular}[t]{@{}l@{}}
            Each\\decomposed \\
            layer
        \end{tabular} & \begin{tabular}[t]{@{}l@{}}
            Transfer final activation \\
            data to edge $a_{i_{\text{noisy}}}$
        \end{tabular} & $n \times b$ \\
        \bottomrule
    \end{tabular}
\end{table}

Note that we chose to keep the SVD-decomposed weight matrix $W_{c_{i}}$ decomposed in the cloud as $U_{c_{i}} \Sigma_{c_{i}} V_{c_{i}}$. This decision is based on the assumption that $k \ll m, n$, which makes the amount of computation lower than performing a full matrix multiplication. Conversely, we decided to reconstruct $W_{e_{i}}$ after decomposition for use on the edge device because in this scenario, using the reconstructed matrix on the edge results in less computation compared to operating with the decomposed form.

\subsection{Total latency calculation}

Since the processing of each phase is sequential and cannot be parallelized, the total latency is the sum of the individual latency of each phase. This ensures that the computation and data transfer steps are accounted for in a linear sequence, reflecting the actual flow of operations during inference.

The total latency $\left(T_{\text {Total }}\right)$ for the purposed protocol, where the model is partitioned between an edge device and a cloud server, can be divided into three main components: edge computations $\left(P^{e}\right)$, cloud computation $\left(P^{c}\right)$, and data transfer $\left(P^{t}\right)$.

Edge Computation are given by
\begin{equation}
    \text{FLOPS}^{\text{edge}} = \left(l - l_{\mathrm{d}}\right) \cdot P^{\mathrm{e}}_{\text{non-dec}} + l_{\mathrm{d}} \cdot P^{\mathrm{e}}_{\text{compute}}
    = 2 m n b \cdot l + n b \cdot \left(l - l_{\mathrm{d}}\right)
\end{equation}
Where $l_{\mathrm{d}}$ represents the decomposed layers and $l$ represents the total layers in the model. With regards to the full model on the edge, we find that 
\begin{equation}
    \text{FLOPS}^{\text{full}} = 2 m n b \cdot l + n b \cdot l
\end{equation}
So the cloud offload consists of $n \cdot l_{\mathrm{d}}$ activation operations.

Cloud Computation are given by
\begin{equation}
    \begin{aligned}
        \text{FLOPS}^{\text{cloud}} &= l_{\mathrm{d}} \cdot \left( P_{\text{compute}}^{\mathrm{c}} + P_{\text{composed}}^{\mathrm{c}} + P_{\text{noise\_gen}}^{\mathrm{c}} + P_{\text{noise\_add}}^{\mathrm{c}} \right) \\
        &= 2 l_{\mathrm{d}} \cdot b \cdot \left( m k + n k + 2 n l_{v} + 1.5 n \right)
    \end{aligned}
\end{equation}

Data Transfer is given by
\begin{equation}
    N^{\text{transfer}} = P_{\text{input}}^{t} + l_{\mathrm{d}} \cdot \left( P_{\text{upload}}^{t} + P_{\text{download}}^{t} \right)
    = n b \left( 2 l_{\mathrm{d}} + 1 \right)
\end{equation}

Total Latency is given by
\begin{equation}
    T_{\text{Total}} = T_{\text{Compute}}\left(\text{FLOPS}^{\text{edge}}\right) + T_{\text{Compute}}\left(\text{FLOPS}^{\text{cloud}}\right) + T_{\text{Transfer}}\left(N^{\text{transfer}}\right)
\end{equation}

\subsection{Selecting parameter values}
Next, we will select specific values for the parameters used in the formulas and analyze the results. This will help illustrate the computational and data transfer latency in our protocol, when the weights are split between an edge device and a cloud server. We based many of our numbers on Llama2 model, treating it as if it were a regular feedforward (FF) network to simplify the analysis and make the estimations more straightforward. Additionally, we assumed a single token generation and chose the one of the decomposition strategies mentioned in the Experiments section as an efficient and robust option.

\begin{table}[H]
    \centering
    \vspace{1em}
    \caption{Summary of model parameters, hardware specifications, and network details}
    \vspace{1em}
    \begin{tabular}{@{}lll@{}}
        \toprule
        \textbf{Parameter} & \textbf{Value} & \textbf{Description} \\
        \midrule
        \multicolumn{3}{l}{\textbf{Model Parameters}} \\
        \midrule
        $l:$ Total Layers & 224 & \begin{tabular}[t]{@{}l@{}}
            7 layers (4 attn + 3 mpl) per block \\
            Llama2 contains 32 blocks
        \end{tabular} \\
        $l_{\mathrm{d}}:$ Decomposed Layers & 70 & \begin{tabular}[t]{@{}l@{}}
            7 layers (4 attn + 3 mpl) per block \\
            10 decomposed blocks in configuration X4
        \end{tabular} \\
        $n, m:$ Matrix Dimensions & 4096, 4096 & Attention Layer \\
        $b:$ batch\_size & 32 & Average input length \\
        $k:$ Secret Singular Values & 50 & X4 configuration \\
        $l_{v}:$ Sampled Noise Vectors & 50 & \\
        $s:$ Activation size in bytes & $4 b$ & FP32 \\
        \midrule
        \multicolumn{3}{l}{\textbf{Hardware Specifications}} \\
        \midrule
        Edge GPU FLOP/sec & 4 TFlops/sec & \begin{tabular}[t]{@{}l@{}}
            Nvidia Jetson Nano spec (estimation for \\
            FP32)
        \end{tabular} \\
        Cloud GPU FLOP/sec & 14 TFlops/sec & Nvidia V100 spec \\
        Utilization & 40\% & Realistic estimation \\
        \midrule
        \multicolumn{3}{l}{\textbf{Network Details}} \\
        \midrule
        $B:$ Network Bandwidth & 25 MB/s & Average bandwidth in the US \\
        $\lambda_{\mathrm{a}}:$ Network Delay & 35 ms & \begin{tabular}[t]{@{}l@{}}
            Average latency in the US \\
            Global IP Network Latency (att.net)
        \end{tabular} \\
        \bottomrule
    \end{tabular}
\end{table}

\subsection{Results}
The results of our calculations are as follows, and are discussed in the paper itself in section 5.
\begin{table}[H]
    \centering
    \vspace{1em}
    \caption{Performance}
    \vspace{1em}
    \begin{tabular}{@{}ll@{}}
        \toprule
        \textbf{Metric} & \textbf{Value} \\
        \midrule
        FLOPs Full Model & 240.547 GFLOPs \\
        FLOPs Edge (Hybrid) & 240.538 GFLOPs \\
        FLOPs Cloud (Hybrid) & 3.697 GFLOPs \\
        Input Transfer & 1.848 M \\
        Compute Edge Latency & 150.34 ms \\
        Compute Cloud Latency & 0.66 ms \\
        Transfer Latency & 71.31 ms \\
        \textbf{Total Latency} & \textbf{222.31 ms} \\
        \bottomrule
    \end{tabular}
\end{table}
\newpage

\subsection{Implications}
It is clear that applying our approach incurs additional latency. In use-cases where the added latency does not significantly impact the user experience, such as offline use-cases (i.e summarization, indexing, processing), our approach may be applied with even higher security, by decomposing additional layers. In use-cases where the computational offload is not as important, such as when the trusted entity is not costly to employ, an even larger number of selected singular components may be used to increase security. However, in scenarios that require real-time responses, our approach should be used sparingly, and only minimal sensitive layers should be selected for decomposition to arrive at feasible latency periods.

\section{Representing a Convolutional Layer as a Fully Connected Layer}
\label{cnn}
For the purpose of decomposing a convolutional layer, we can represent the convolutional layer as a fully connected layer, and then perform SVD to decompose the layer as detailed in our paper.
\textbf{Convolutional Layer}. Consider the input tensor $\mathbf{X}$ of shape $(H, W, C)$, the convolutional kernel $\mathbf{K}$ of shape $(kH, kW, C, N)$, and the output tensor $\mathbf{O}$ of shape $(H_o, W_o, N)$. For simplicity, assume stride $s = 1$ and no padding.

The output at position $(i, j)$ for filter $n$ is given by:
\[
\mathbf{O}_{i,j,n} = \sum_{c=1}^{C} \sum_{u=1}^{kH} \sum_{v=1}^{kW} \mathbf{X}_{i+u-1,j+v-1,c} \cdot \mathbf{K}_{u,v,c,n}
\]

\textbf{Fully Connected Layer}. The equivalent fully connected layer will have an input vector $\mathbf{x}$ and a weight matrix $\mathbf{W}$. The input tensor $\mathbf{X}$ is unrolled into a vector $\mathbf{x}$ of size $H \cdot W \cdot C$. The unrolling process is detailed as follows.

1. \textbf{Flatten the Input:}
\[
\mathbf{x} = \text{vec}(\mathbf{X})
\]

2. \textbf{Construct the Weight Matrix:}
   The weight matrix $\mathbf{W}$ for the fully connected layer will be of shape $(H_o \cdot W_o \cdot N, H \cdot W \cdot C)$.
   Each row of $\mathbf{W}$ corresponds to a specific position of the kernel applied to the flattened input.

For each output position $(i, j)$ and filter $n$:
\[
\mathbf{W}_{n, \text{index}(i,j,c)} = \mathbf{K}_{u,v,c,n}
\]
where $\text{index}(i,j,c)$ maps the 3D position $(i,j,c)$ in the input tensor to the corresponding position in the flattened vector.

\subsection*{Analytical Formulation}

For each output position $(i, j)$ and filter $n$:
\[
\mathbf{O}_{i,j,n} = \sum_{c=1}^{C} \sum_{u=1}^{kH} \sum_{v=1}^{kW} \mathbf{X}_{i+u-1,j+v-1,c} \cdot \mathbf{K}_{u,v,c,n}
\]

Unroll $\mathbf{X}$ to $\mathbf{x}$ and define the corresponding weight vector $\mathbf{w}_{(i,j,n)}$ in $\mathbf{W}$:
\[
\mathbf{w}_{(i,j,n)} = \text{vec}(\mathbf{K}_{:,:,c,n})
\]

The output $\mathbf{O}$ is obtained by matrix multiplication:
\[
\mathbf{o} = \mathbf{W} \cdot \mathbf{x}
\]
where $\mathbf{o}$ is the flattened version of $\mathbf{O}$.

\label{app:hybrid_inference_protocols}

\section{Secure Hybrid Inference Protocol for Attention Layers}
\label{app:attention}
In the vanilla Transformer model, each attention head receives a sequence of tokens embedding denoted by $\boldsymbol{X} \in \mathbb{R}^{n \times d}$, where $n$ represents the sequence length, $d$ is the embedding dimension, and $d_h$ is the number of tokens in a batch. The weight parameters $\boldsymbol{W}_q, \boldsymbol{W}_v \in \mathbb{R}^{d \times d_h}$, and $\boldsymbol{W}^V \in \mathbb{R}^{d \times d_h}$ are used to transform the input features $\boldsymbol{X}$ into query $\boldsymbol{Q}$, key $\boldsymbol{K}$, and value $\boldsymbol{V}$, respectively. 
$$
\begin{gathered}
\boldsymbol{Q}=\boldsymbol{X} \boldsymbol{W}_q, \quad \boldsymbol{K}=\boldsymbol{X} \boldsymbol{W}_k, \quad \boldsymbol{V}=\boldsymbol{X} \boldsymbol{W}_v, \\
\operatorname{Attn}(\boldsymbol{X})=\text{Softmax}\left(\frac{\boldsymbol{Q} \boldsymbol{K}^T}{\sqrt{d_h}}\right) \boldsymbol{V}.
\end{gathered}
$$

In the following, we provide the secure inference protocol for the attention head layer.  
\begin{figure}[H]
    \centering
     \scalebox{0.7}{
        \fbox{%
            \procedure{\textsc{Attention Head - Secure Inference Protocol} (On-line,  per inference instance)}{%
                \< \< \\
                \textbf{$\charlie$harlie} \<   \<  \textbf{$\david$avid} \\
                \< \< \\
                {%
                    \begin{subprocedure}\procedure{\textbf{Inner Masking}}{%
                     \bullet \text{  Input to this layer: } a_{i-1} \\
                     \bullet \text{  Retrieve pre-computed random masks } r^q_{i-1}, r^k_{i-1} \\
                     \bullet \quad  \textcolor{blue}{\widetilde{a^q}_{i-1} := a_{i-1} + r^q_{i-1}} \\ 
                     \bullet \quad  \textcolor{blue}{\widetilde{a^k}_{i-1} := a_{i-1} + r^k_{i-1}} \\ 
                     \bullet \quad {a_q}_i^{\charlie} := {W_q}_i^{\charlie} a^q_{i-1}\\ 
                     \bullet \quad {a_k}_i^{\charlie} := {W_k}_i^{\charlie} a^k_{i-1} 
                    }
                    \end{subprocedure}
                } \<  \< \\
                  \< \textcolor{blue}{\sendmessageright*[1.5 cm]{\widetilde{a}_{i-1}}}  \<   \\
                  \<  \< 
                {%
                    \begin{subprocedure}\procedure{\textbf{$\david$'s Local Compute}}{%
                    \bullet \quad \widetilde{{a_q}_i^{\david}} := {W_q}_i^{\david} \cdot  \widetilde{a^q}_{i-1}\\
                    \bullet \quad \widetilde{{a_k}_i^{\david}} := {W_k}_i^{\david} \cdot  \widetilde{a^k}_{i-1}\\
                    }
                    \end{subprocedure}
                }\\
               \< \textcolor{blue}{\sendmessageleft*[1.5 cm]{\widetilde{{a_q}_i^{\david}},\widetilde{{a_k}_i^{\david}}}} \<   \\
               {%
                    \begin{subprocedure}\procedure{\textbf{Inner Unmasking}}{%
                    \bullet \text{Retrieve pre-computed cancellation mask $c_{q_i},c_{k_i}$ } \\
                    \bullet \quad \textcolor{blue}{{a_q}_i^{\david} = \widetilde{{a_q}_i^{\david}} - c_{q_i} = W_{q_i}^{\david} \cdot \left(\widetilde{a_q}_{i-1} -  r^q_{i-1} \right) = W_{q_i}^{\david} \cdot {a}_{i-1}}\\
                    \bullet \quad \textcolor{blue}{{a_k}_i^{\david} = \widetilde{{a_k}_i^{\david}} - c_{k_i} = W_{k_i}^{\david} \cdot \left(\widetilde{{a_k}}_{i-1} -  r^k_{i-1} \right) = W_{k_i}^{\david} \cdot {a}_{i-1}}
                    }
                    \end{subprocedure}
                } \<   \<  \\
                \< \< \\
               {%
                    \begin{subprocedure}\procedure{\textbf{Compute Midterm Softmax Output}}{%
                    \bullet  \quad a_{s_i} = \operatorname{softmax} \left(\frac{\left({a_q}_i^{\charlie}+{a_q}_i^{\david}\right) \cdot\left({a_k}_i^{\charlie}+({a_k}_i^{\david}\right)^T}{\sqrt{|a_{i-1}|}}\right) = \operatorname{softmax}\left(\frac{Q K^T}{\sqrt{|a_{i-1}|}} \right) 
                    }
                    \end{subprocedure}
                } \<   \<  \\
{%
                    \begin{subprocedure}\procedure{\textbf{Outer Masking}}{%
                     \bullet \text{  Retrieve pre-computed random mask } r^v_{i-1} \\
                     \bullet \quad  \textcolor{blue}{\widetilde{a}_{s_i}  := {a}_{s_i}  + r^v_{i-1}} \\ 
                     \bullet \quad a_i^{\charlie} := W_{v_i}^{\charlie} {a}_{s_i} 
                    }
                    \end{subprocedure}
                } \<  \< \\
                  \< \textcolor{blue}{\sendmessageright*[1.5 cm]{\widetilde{a}_{s_i}}}  \<   \\
                  \<  \< 
                {%
                    \begin{subprocedure}\procedure{\textbf{$\david$'s Local Compute}}{%
                    \bullet \quad \widetilde{a_i^{\david}} := W_{v_i}^{\david} \cdot \widetilde{a}_{s_i}
                    }
                    \end{subprocedure}
                }\\
               \< \textcolor{blue}{\sendmessageleft*[1.5 cm]{\widetilde{a_i^{\david}}}} \<   \\
               {%
                    \begin{subprocedure}\procedure{\textbf{Outer Unmasking}}{%
                    \bullet \text{ Retrieve pre-computed cancellation mask $c_{v_i}$ } \\
                    \bullet \quad \textcolor{blue}{a_i^{\david} = \widetilde{a_i^{\david}} - c_{v_i} = W_{v_i}^{\david} \cdot \left(\widetilde{a}_{s_i} -  r^v_{i-1} \right) = W_{v_i}^{\david} \cdot {a}_{s_i}}
                    }
                    \end{subprocedure}
                } \<   \<  \\
                \< \< \\
               {%
                    \begin{subprocedure}\procedure{\textbf{Compute Attention Head Output}}{%
                    \bullet  \quad a_i = \sigma (a_i^{\charlie} + a_i^{\david}) = \sigma ((W_{v_i}^{\charlie} + W_{v_i}^{\david}) a_{i-1}) = \sigma (W_{v_i} a_{i-1}) \\
                    \bullet  \quad a_i = \sigma (a_i^{\charlie} + a_i^{\david}) = \sigma ((W_i^{\charlie} + W_i^{\david}) a_{i-1}) = \sigma (W_i a_{i-1}) 
                    }
                    \end{subprocedure}
                } \<   \<  \\
            } 
         } 
    }  
    \caption{Double Step Secure Inference Protocol For Attention Head  -- Online, per Inference Instance}
    \label{fig:attention}
\end{figure}


